\documentclass[11pt]{article}

\newcommand{\aop}{Ann. Phys.~}
\newcommand{\cmp}{Comm. Math. Phys.~}

\newcommand{\jmp}{J. Math. Phys.~}
\newcommand{\jpa}{J. Phys. A~}

\newcommand{\prl}{Phys. Rev. Lett.~}
\newcommand{\pra}{Phys. Rev. A~}
\newcommand{\pla}{Phys. Lett. A~}

\usepackage{color}
\usepackage{subfigure}
\usepackage[T1]{fontenc}
\usepackage[sc]{mathpazo}
\usepackage{amsmath}
\usepackage{amssymb}
\usepackage{graphicx,color}
\usepackage{framed}
\usepackage{multirow}
\usepackage{enumerate}
\usepackage{amsthm}
\usepackage{amsfonts,mathrsfs}
\usepackage{geometry} 
\usepackage{eepic}
\usepackage{ifthen}
\usepackage[vcentermath]{youngtab}
\usepackage[unicode=true,pdfusetitle, bookmarks=true,bookmarksnumbered=false,bookmarksopen=false, breaklinks=false,pdfborder={0 0 0},backref=false,colorlinks=false] {hyperref}
\hypersetup{
colorlinks,linkcolor=myurlcolor,citecolor=myurlcolor,urlcolor=myurlcolor}
\definecolor{myurlcolor}{rgb}{0,0,0.7}

\usepackage{color}

\usepackage{hyperref}
\hypersetup{pdfpagemode=UseNone}

\newcommand{\tinyspace}{\mspace{1mu}}

\newcommand{\abs}[1]{\left\lvert\tinyspace #1 \tinyspace\right\rvert}

\renewcommand{\t}{{\scriptscriptstyle\mathsf{T}}}

\newcommand{\im}{\op{im}}

\def \dif {\mathrm{d}}
\def \diag {\mathrm{diag}}

\def \im {\mathrm{Im}}

\def\complex{\mathbb{C}}
\def\real{\mathbb{R}}

\def\I{\mathbb{1}}

\newenvironment{mylist}[1]{\begin{list}{}{
    \setlength{\leftmargin}{#1}
    \setlength{\rightmargin}{0mm}
    \setlength{\labelsep}{2mm}
    \setlength{\labelwidth}{8mm}
    \setlength{\itemsep}{0mm}}}
    {\end{list}}

\def\ot{\otimes}

\newcommand{\out}[2]{| #1\rangle\langle #2 |}

\newcommand{\Innerm}[3]{\left\langle #1 \left| #2 \right| #3 \right\rangle}

\newcommand{\Pa}[1]{\left(#1\right)}

\newcommand{\Br}[1]{\left[#1\right]}

\newcommand{\Set}[1]{\left\{#1\right\}}

\newcommand{\ket}[1]{|#1\rangle}

\DeclareMathOperator{\trace}{Tr}

\newcommand{\Ptr}[2]{\trace_{#1}\Pa{#2}}

\newcommand{\Tr}[1]{\Ptr{}{#1}}

\def\cH{\mathcal{H}}\def\cI{\mathcal{I}}

\def\cU{\mathcal{U}}

\def\bC{\mathbf{C}}
\def\bG{\mathbf{G}}

\def\bP{\mathbf{P}}\def\bQ{\mathbf{Q}}

\def\bsC{\boldsymbol{C}}
\def\bsH{\boldsymbol{H}}

\def\bsQ{\boldsymbol{Q}}\def\bsS{\boldsymbol{S}}

\def\bsb{\boldsymbol{b}}

\def\rC{\mathrm{C}}
\def\rG{\mathrm{G}}\def\rH{\mathrm{H}}

\def\rP{\mathrm{P}}

\newtheorem{thrm}{Theorem}[section]
\newtheorem{lem}[thrm]{Lemma}

\newtheorem{cor}[thrm]{Corollary}
\theoremstyle{definition}

\newtheorem{exam}[thrm]{Example}

\numberwithin{equation}{section}

\newcounter{questionnumber}

\begin{document}

\title{Coherence generating power of unitary transformations \\ via probabilistic average}

\author{Lin Zhang$^1$\footnote{E-mail: godyalin@163.com; linyz@hdu.edu.cn},\quad Zhihao Ma$^2$,\quad Zhihua Chen$^3$,\quad Shao-Ming Fei$^4$\footnote{feishm@cnu.edu.cn}\\
  {$^1$\it\small Institute of Mathematics, Hangzhou Dianzi University, Hangzhou 310018, PR~China}\\
  {$^2$\it\small Department of Mathematics, Shanghai Jiaotong University, Shanghai 200240,
  PR~China}\\
  {$^3$\it\small Department of Applied Physics, Zhejiang University of Technology, Hangzhou, Zhejiang 310023, PR~China}\\
  {$^4$\it\small School of Mathematical Sciences, Capital Normal University, Beijing 100048, PR~China}}

\date{}
\maketitle
\begin{abstract}

We study the ability of a quantum channel to generate quantum
coherence when it applies to incoherent states. Based on
probabilistic averages, we define a measure of such coherence
generating power (CGP) for a generic quantum channel, based on the
average coherence generated by the quantum channel acting on a
uniform ensemble of incoherent states. Explicit analytical formula
of the CGP for any unitary channels are presented in terms of
subentropy. An upper bound for CGP of unital quantum channels has
been also derived. Detailed examples are investigated.

\end{abstract}

\newpage

\section{Introduction}

Originating from the fundamental superposition principle of quantum
mechanics, quantum coherence is a kind of important quantum
resources. It plays key roles in the interference of light, the
laser, superconductivity and quantum thermodynamics
\cite{Mandel,London,Horodecki}, as well as in some quantum
information tasks \cite{Bagan,Jha,Kammerlander,Shi} and biological
processes \cite{Lloyd,Li,Huelga,Singh0}.
However, the rigorous theories of quantum coherence have been
proposed only recently \cite{Baumgratz}. While the rigorous
characterization of the superposition in terms of resource theory
appeared even late \cite{Theurer2017}, although the idea of
measuring the degree of superposition in quantum states had been introduced
early in \cite{Aberg2006}.

The coherence measures are provided to quantify
the amount of quantum coherence for a given quantum system.
{After the work of Baumgratz \emph{et al.} \cite{Baumgratz},
various aspects of coherence have been studied in the literature.}
Recently, many different kinds of coherence measures such as
coherence of formation, relative entropy of coherence, $l_1$ norm of
coherence, distillable coherence, robustness of coherence, coherence
averaged over all basis sets or the Haar distributed pure states,
and max-relative entropy of coherence have been investigated
\cite{Baumgratz,Winter,Cheng,ROC,Singh,aop2017,bukf}. {The
notion of speakable and unspeakable coherence is discussed in
\cite{Marvian2016}.}

Based on these measures of coherence, the connections of coherence
with path distinguishability and asymmetry have been studied
\cite{Piani,Marvian}. For bipartite and multipartite systems, the
relationship between quantum coherence and other quantum
correlations such as quantum entanglement and quantum discord has
also been studied
\cite{Streltsov,Radhakrishnan,Ma,Karpat,Malvezzi,Chitambar}. It has
been shown that there is a one to one mapping between the quantum
entanglement and quantum coherence \cite{Zhu}.

{Apart from the above investigations, Mani and
Karimipour \cite{Mani} first introduced the concept of
cohering power and de-cohering power of generic quantum channels. They
defined the coherence generating power (CGP) of a
quantum channel to quantify the power of a channel in
generating quantum coherence by
optimizing the output coherence. And several examples of
qubit channels including unitary gates are presented}. Different
kinds of operations which can either preserve or generate coherence
have been also studied \cite{Misra,MGD}. Probabilistic averages were
firstly used to study the CGP by Zanardi \emph{et al.}
\cite{Zanardi1,Zanardi2}. They presented a way to quantify the CGP
of a unitary gates, by introducing a measure based on the average
coherence generated by the channel acting on a uniform ensemble of
incoherent states. In deriving explicit analytical formulae of CGP
for any dimensional systems, they used the Hilbert-Schmidt norm as
a measure of coherence.

However, the Hilbert-Schmidt norm measure is not a bona fide
measure of coherence. It does not have the
desired monotonicity property in general, although it facilitates
the calculation of CGP. In the present paper we use the relative entropy coherence measure,
which is a well defined measure of coherence and satisfies all the
required properties of a bona fide measure of coherence, together
with informationally operational implications. We use the relative
entropy of coherence to quantify the CGP of a generic quantum
channel via probabilistic averages. We give an explicit analytical
formula of CGP for any unitary channels. An upper bound for CGP of a
unital quantum channel is also derived.

\section{CGP of quantum channels}

The measure of coherence under consideration in the present paper is
the relative entropy of coherence \cite{Baumgratz}:
\begin{eqnarray}
\bsC_r(\rho)=\bsS(\rho_\diag)-\bsS(\rho),
\end{eqnarray}
where $\bsS(\rho)=-\Tr{\rho\ln\rho}$ is the von Neumann entropy of a
quantum state $\rho$ and $\rho_\diag$ is the diagonal part of $\rho$
with respect to the standard basis. Through out the paper, we take
$\Set{\ket{i}:i=1,\ldots,N}$ the standard computational basis in an
$N$-dimensional Hilbert space $\cH_N$. Denote $\cI$ the set of
incoherent states with respect to the basis. An incoherent state
$\Lambda$ in $\cI$ has the form
$\Lambda=\diag(\lambda_1,\ldots,\lambda_N)$, where
$\lambda=(\lambda_1,\ldots,\lambda_N)$ constitutes an
$N$-dimensional probability vector with $\sum_{i=1}^N \lambda_i=1$.
Obviously $\bsC_r(\Lambda)=0$. The problem one may ask is that if
$\Lambda$ undergoes a generic quantum channel $\Phi$, i.e., a
trace-preserving completely positive and linear map, what the
coherence of $\Phi(\Lambda)$ will be.

To characterize the coherence generating power of a generic quantum
channel $\Phi$, one needs to average over all the incoherent states
$\Lambda$. Nevertheless, the definition of CGP of a quantum channel
is not unique. All current approaches provided involve optimizations
problems that are extremely hard to deal with for generic channels.
By adopting the probabilistic averages \cite{Zanardi1,Zanardi2}, we
define the coherence generating power $\bC\bG\bP(\Phi)$ of $\Phi$ to
be
\begin{eqnarray}\label{CGB}\notag
\bC\bG\bP(\Phi)&:
=&\int_\cI\dif\mu(\Lambda)\bsC_r(\Phi(\Lambda))\\
&=&\int_\cI\dif\mu(\Lambda)\Br{\bsS(\Phi(\Lambda)_\diag) -
\bsS(\Phi(\Lambda))},
\end{eqnarray}
where $\dif\mu(\lambda) =
\Gamma(N)\delta\Pa{1-\sum^N_{j=1}\lambda_j}\prod^N_{j=1}\dif\lambda_j$,
i.e., $\mu$ is the probability measure on a uniform ensemble of
incoherent states.

We first calculate the $\bC\bG\bP(\Phi)$ for unitary channels
$\Phi=\mathrm{Ad}_U$ such that $\Phi(\Lambda)=U\Lambda{U^\dag}$,
where $U$ denotes unitary transformations and $\dag$ the transpose
and conjugation. Before giving the main results, we introduce some
basic notations. Let $p=[p_1,\ldots,p_N]^\t$ and
$q=[q_1,\ldots,q_N]^\t$ be two probability vectors in $\real^N$,
where $^\t$ denotes the transpose. The \emph{Shannon entropy} of $p$
and the \emph{relative entropy} of $p$ and $q$ are defined by
$\bsH(p)=-\sum^N_{i=1}p_i\ln p_i$ and
$\bsH(p||q)=\sum^N_{i=1}p_i(\ln p_i - \ln q_i)$, respectively, where
$0\ln0=0$.

An $N\times N$ matrix $B=[b_{ij}]$ is said to be \emph{stochastic}
if $b_{ij}\geqslant0$, and $\sum^N_{i=1}b_{ij}=1$ for every
$j=1,\ldots,N$. If $\sum^N_{j=1}b_{ij}=1$ holds also for every
$i=1,\ldots,N$, then a stochastic $B$ is said to be
\emph{bi-stochastic}. Let $B$ be a bi-stochastic $N\times N$ matrix
and $p$ an $N$-dimensional probability vector. The \emph{weighted
entropy of $B$ with respect to $p$} is defined by $\bsH_p(B) =
\sum^N_{j=1}p_j\bsH(\beta_j)$, where $B=[\beta_1,\ldots,\beta_N]$ is
the column-block partition of $B$. In particular, when
$p=[1/N,\ldots,1/N]^\t$, one denotes
\begin{eqnarray}
\bsH(B)=\frac1N\sum^N_{j=1}\bsH(\beta_j).
\end{eqnarray}
It can be  proved that
$\bsH_p(B)\leqslant \bsH(Bp)\leqslant \bsH_p(B)+\bsH(p)$.

Let $\Phi$ be a quantum channel and $\Phi=\sum_\mu
\mathrm{Ad}_{M_\mu}$ be its Kraus representation. Define the
\emph{Kraus matrix} $B(\Phi)$ of $\Phi$ by $B(\Phi)=\sum_\mu
M_\mu\star \overline{M_\mu}$, where $\star$ denotes the Schur
product of matrices, that is, the entrywise product of two matrices,
and $\overline{M_\mu}$ is the complex conjugate of $M_\mu$. It is
easy to show that $B(\Phi)$ is a stochastic matrix if $\Phi$ is a
quantum channel on $\cH$, and $B(\Phi)$ is a bi-stochastic matrix if
$\Phi$ is a unital quantum channel ($\Phi$ being unital here means
that $\Phi(\I)=\I$). Moreover, $B(\Phi^\dagger)=B(\Phi)^\t$
\cite{LJ2011}. In this case, one also has $p=B(\Phi)\lambda$, where
$p=[p_1,\ldots,p_N]^\t$ with $p_j=\Innerm{j}{\Phi(\rho)}{j}$,
$j=1,\ldots,N$, and $\lambda=[\lambda_1,\ldots,\lambda_N]^\t$ with
$\lambda_i$ giving by the spectral decomposition
$\rho=\sum^N_{j=1}\lambda_j\out{j}{j}$ of a quantum state $\rho$.

If $B=[b_{ij}]$ is a $N\times N$ bi-stochastic matrix and
$\lambda=[\lambda_1,\ldots,\lambda_N]^\t$ a probability vector, then
$B\lambda$ is also a probability vector. Its Shannon entropy is
given by $\bsH(B\lambda)$. It is well-known that the action of
bi-stochastic $B$ on probability vectors increases the uncertainty,
i.e. $\bsH(B\lambda) \geqslant \bsH(\lambda)$ --- a fact for the
first step in proving the famous $H$-theorem \cite{Lasota1994}. With
respect to a random probability vector $\lambda$ subjecting to a
uniform distribution over the probability simplex
$\Delta_{N-1}=\Set{[x_1,\ldots,x_N]\in\real^N_+:
\sum^N_{j=1}x_j=1}$, the corresponding probability measure
$\dif\mu(\lambda)$ is given by the one in (\ref{CGB}). Moreover, the
\emph{subentropy} associated with $\lambda$ is defined by
\begin{eqnarray}\label{se}
\bsQ(\lambda) = -\sum^N_{i=1} \frac{\lambda^N_i\ln
\lambda_i}{\prod_{j\neq i}(\lambda_i-\lambda_j)},
\end{eqnarray}
which takes its maximal value $\bsQ(\I_N/N)=\ln N- H_N + 1$ for the
completely mixed states, where $H_N=\sum^N_{j=1}1/j$ is the $N$-th
harmonic number \cite{Page1993,lin2017}.

Similarly, we can define weighted subentropy of a stochastic matrix
$B$ with respect to a probability vector $p$, $\bsQ_p(B) =
\sum^N_{j=1}p_j\bsQ(\beta_j)$, where $B=[\beta_1,\ldots,\beta_N]$ is
the column-block partition of $B$. In particular, when
$p=[1/N,\ldots,1/N]^\t$, we denote
\begin{eqnarray}
\bsQ(B)=\frac1N\sum^N_{j=1}\bsQ(\beta_j).
\end{eqnarray}
The explicit formula of CGP for the unitary channels can be given by
the subentropy.

\section{CGP of unitary and unital channels}

Based on the definition of CGP of a quantum channel, we may derive
an explicit analytical formula of the CGP for any unitary channels.

\begin{thrm}\label{th:Aone}
For any given $N\times N$ unitary matrix $U$, the $\bC\bG\bP$ of the
unitary channel $\mathrm{Ad}_U$ is given by
\begin{eqnarray}
\bC\bG\bP(U) = \bsQ(B(U)^\t),
\end{eqnarray}
where $B(U):=B(\mathrm{Ad}_U)= U\star \overline{U}$.
\end{thrm}

Before proving the theorem, we first give the following Lemma.

\begin{lem}
Let $B$ be an $N\times N$ bi-stochastic matrix. Then
$\int\bsH(B\lambda)\dif\mu(\lambda) = H_N-1+\bsQ(B^\t)$.
Furthermore,
\begin{eqnarray}\label{lemma}
\int[\bsH(B\lambda)-\bsH(\lambda)]\dif\mu(\lambda) = \bsQ(B^\t).
\end{eqnarray}
\end{lem}

\begin{proof}
We calculate the following integrals related to the left hand side
of (\ref{lemma}):
\begin{eqnarray*}
\cI_B=\int\bsH(B\lambda)\dif\mu(\lambda)\quad\text{and}\quad
\cI_{\I}=\int\bsH(\lambda)\dif\mu(\lambda).
\end{eqnarray*}
Concerning $\cI_B$, we have
\begin{eqnarray}
\begin{aligned}
\cI_B=\int\bsH(B\lambda)\dif\mu(\lambda)=
-\sum^N_{i=1}\int\Pa{\sum^N_{j=1}b_{ij}\lambda_j}\ln\Pa{\sum^N_{j=1}b_{ij}\lambda_j}\dif\mu(\lambda).\nonumber
\end{aligned}
\end{eqnarray}
It suffices to calculate
\begin{eqnarray*}
\Gamma(N)\int\Pa{\sum^N_{j=1}p_j\lambda_j}\ln\Pa{\sum^N_{j=1}p_j\lambda_j}\delta\Pa{1-\sum^N_{j=1}\lambda_j}\prod^N_{k=1}\dif\lambda_k
=\cI'_p(1),
\end{eqnarray*}
where
\begin{equation}\label{ipa}
\cI_p(\alpha)=\Gamma(N)\int
\Pa{\sum^N_{j=1}p_j\lambda_j}^\alpha\delta\Pa{1-\sum^N_{j=1}\lambda_j}\prod^N_{k=1}\dif\lambda_k
\end{equation}
and $\cI'_p(1)
=\left.\frac{d\cI_p(\alpha)}{d\alpha}\right|_{\alpha=1}$. After some
tedious calculation , we have (see Eq.~\eqref{eq:alpha-int} in
Appendix A),
\begin{eqnarray}
\cI_p(\alpha)=\frac{\Gamma(N)\Gamma(\alpha+1)}{\Gamma(\alpha+N)}\sum^N_{j=1}\frac{p^{\alpha+N-1}_j}{\prod_{i\neq
j}(p_j-p_i)}\nonumber
\end{eqnarray}
and (see Eq.~\eqref{eq:prime-at-one} in Appendix A)
\begin{eqnarray}
\cI'_p(1) = -\frac1N\Pa{H_N-1+\bsQ(p)}.
\end{eqnarray}
By partitioning $B$ as a row-block matrix:
$$
B=\Br{\begin{array}{c}
                                                 \bsb_1 \\
                                                 \vdots \\
                                                 \bsb_N
                                               \end{array}
},
$$
where $\bsb_i=[b_{i1},\ldots,b_{iN}]$ for $i=1,\ldots,N$, we obtain
\begin{eqnarray}\label{8}
\cI_B = -\sum^N_{i=1}\cI'_{\bsb^\t_i}(1) =H_N-1+\frac1N\sum^N_{i=1}
\bsQ(\bsb^\t_i).
\end{eqnarray}
Taking $B=\I$, we have $\cI_{\I}=H_N-1$, which gives rise to
(\ref{lemma}).
\end{proof}

\textbf{Remark} It can shown that $\bsQ(B^\t) \leqslant \bsH(B)$,
see Appendix B. Hence (\ref{lemma}) also implies that
$\int[\bsH(B\lambda)-\bsH(\lambda)]\dif\mu(\lambda) \leqslant
\bsH(B)$.

\begin{proof}[Proof of Theorem~\ref{th:Aone}]
Let $\Lambda=\diag(\lambda_1,\ldots,\lambda_N)$ be an incoherent
state in $\cI$, and $\Phi=\mathrm{Ad}_U$ be a unitary channel.
Denote $\lambda=(\lambda_1,\ldots,\lambda_N)$ the probability vector
form of $\Lambda$. Then
$$
\bsS((U\Lambda U^\dagger)_\diag) =
\bsH(B(\Phi)\lambda)\quad\text{and}\quad \bsS(U\Lambda U^\dagger) =
\bsH(\lambda).
$$
Thus $\bsS((U\Lambda U^\dagger)_\diag) - \bsS(U\Lambda U^\dagger) =
\bsH(B(\Phi)\lambda) - \rH(\lambda)$. Therefore
\begin{eqnarray*}
\Gamma(N)\int[\dif\Lambda]
(1-\Tr{\Lambda})\Pa{\bsS(\Phi(\Lambda)_\diag) - \bsS(\Lambda)}
=\int\dif\mu(\lambda)\Pa{\bsH(B(\Phi)\lambda) - \bsH(\lambda)}.
\end{eqnarray*}
That is,
\begin{eqnarray*}
\bC\bG\bP(\Phi) &=& \int\dif\mu(\lambda)\Pa{\bsH(B(\Phi)\lambda)
-\bsH(\lambda)}\\
&=&\frac1N\sum^N_{i=1}\bsQ(\bsb^\t_i(\Phi))=\bsQ(B(U)^\t).
\end{eqnarray*}
We have done.
\end{proof}

From the Theorem we see that the possible values of CGP form the
closed interval $[0,\ln N-H_N+1]$. An interesting question is which
kind of unitary channels would give rise to the maximal value of
CGP. Let us consider the set of $U$ such that
\begin{eqnarray}
\Set{U:\bC\bG\bP(U)=\ln N-H_N+1} = \Set{U: B(U)=\frac1N P},
\end{eqnarray}
where $P$ is the matrix with all entries being one. Obviously $U$
must be of the following form: $U =\frac1{\sqrt{N}}Z$, where
$Z=[z_{ij}]$ with the complex entries $z_{ij}$ satisfying
$\abs{z_{ij}}=1$. For example, for $N=2$, we have
\begin{eqnarray}
U=\frac1{\sqrt{2}}e^{\mathrm{i}\phi}\Br{\begin{array}{cc}
                                          e^{\mathrm{i}\theta} & -e^{-\mathrm{i}\gamma} \\
                                          e^{\mathrm{i}\gamma} & e^{-\mathrm{i}\theta}
                                        \end{array}
}.
\end{eqnarray}

If $\Phi$ is a unital quantum channel, one has
\begin{eqnarray}
\bsS(\Phi(\rho)) - \bsS(\rho)\geqslant
\bsS(\rho||\Phi^*\circ\Phi(\rho)),
\end{eqnarray}
where
$\bsS(\rho||\sigma):=-\Tr{\rho(\ln\rho-\ln\sigma)}$ is the relative
entropy, and $\Phi^*$ is the dual of $\Phi$ in the sense that
$\Tr{X\Phi^*(Y)} = \Tr{\Phi(X)Y}$ for any $N\times N$ matrices $X$
and $Y$ \cite{buscemi2016pra}. In this case we have

\begin{cor}
If $\Phi$ is a unital quantum channel, then
\begin{eqnarray}
\bC\bG\bP(\Phi) \leqslant \bsQ(B(\Phi)^\t),
\end{eqnarray}
where $B(\Phi)$ is the Kraus matrix of $\Phi$.
\end{cor}

\section{Examples}

In the following, as applications of our Theorem~\ref{th:Aone}, we
calculate the CGP for some detailed unitary transformations.

\begin{exam}
Consider the Hadamard gate $H=\frac1{\sqrt{2}}\Br{\begin{array}{cc}
                                         1 & 1 \\
                                         1 & -1
                                       \end{array}}$.
The Kraus matrix is given by $B(H)=\frac12\Br{\begin{array}{cc}
                                         1 & 1 \\
                                         1 & 1
                                       \end{array}}$.
Therefore, from the Theorem we have $\bC\bG\bP(H) = \ln2  - 1/2$.
\end{exam}

\begin{exam}
For $U_{\theta}=\Br{\begin{array}{cc}
                             \cos\theta & \sin\theta \\
                             -\sin\theta & \cos\theta
                             \end{array}}$,
the Kraus matrix is given by $B(U_\theta)=\Br{\begin{array}{cc}
                                        \cos^2\theta  & \sin^2\theta \\
                                         \sin^2\theta & \cos^2\theta
                                       \end{array}}$.
Its CGP is given by
\begin{eqnarray}
\rC\rG\rP(U_{\theta}) =
\frac{\sin^4\theta\ln\sin^2\theta-\cos^4\theta\ln\cos^2\theta}{\cos^2\theta-\sin^2\theta}.
\end{eqnarray}

\begin{figure}[htbp]\centering
\includegraphics[width=0.5\textwidth]{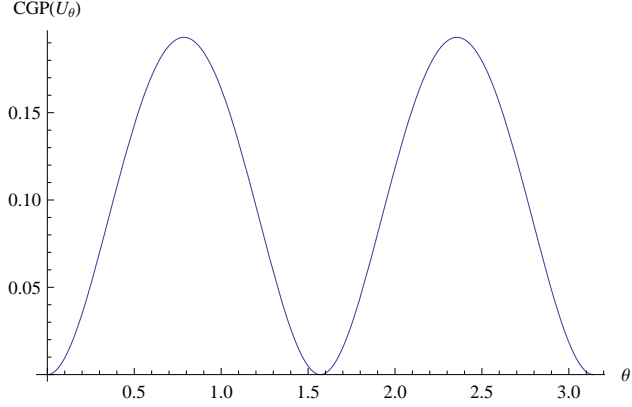}
\caption{The coherence generating power of $U_{\theta}$ vs
$\theta$}. \label{fig:Fig1}
\end{figure}
As a demonstration, we plot the $\rC\rG\rP(U_{\theta})$ as the
function of $\theta\in[0,\pi]$. From Fig.~\ref{fig:Fig1}, we see
that the coherence generating power of $U_\theta$ is a periodic
function of $\theta$. In particular, the maximal CGP for $U_\theta$
is $\ln N-H_N+1=\ln 2-1/2=0.193$. We also see that the maximal CGP
of $U_{\theta}$ is attained at $\theta=\pi/4$ and $3\pi/4$.
\end{exam}

\begin{exam} [Square root of swap gate]
The $\sqrt{\text{swap}}$ gate is universal in the sense that any
quantum multi-qubit gates can be constructed from
$\sqrt{\text{swap}}$ and single qubit gates,
$$
\sqrt{\text{swap}} = \Br{\begin{array}{cccc}
                           1 & 0 & 0 & 0 \\
                           0 & \frac12(1+\mathrm{i}) & \frac12(1-\mathrm{i}) & 0 \\
                           0 & \frac12(1-\mathrm{i}) & \frac12(1+\mathrm{i}) & 0 \\
                           0 & 0 & 0 & 1
                         \end{array}
}.
$$
By direct computation we have
$$
\bC\bG\bP(\sqrt{\text{swap}}) = \frac12\ln 2.
$$
\end{exam}

\begin{exam}
For a partial swap operator \cite{Audenaert2016}, one has $U_t\in
\cU(\complex^d\ot\complex^d)$:
$U_t=\sqrt{t}\mathbb{I}_d\ot\mathbb{I}_d+\mathrm{i}\sqrt{1-t}\,S$,
where $S=\sum^d_{i,j=1}\out{ij}{ji}$ and $t\in[0,1]$. In particular,
for $d=2$, we have
$$
U_t = \Br{\begin{array}{cccc}
            \sqrt{t}+\sqrt{1-t}\mathrm{i} & 0 & 0 & 0 \\
            0 & \sqrt{t} & \sqrt{1-t}\mathrm{i} & 0 \\
            0 & \sqrt{1-t}\mathrm{i} & \sqrt{t} & 0 \\
            0 & 0 & 0 & \sqrt{t}+\sqrt{1-t}\mathrm{i}
          \end{array}
}.
$$
Then
$$
B(U_t) = U_t\star \overline{U}_t=\Br{\begin{array}{cccc}
            1 & 0 & 0 & 0 \\
            0 & t & 1-t & 0 \\
            0 & 1-t & t & 0 \\
            0 & 0 & 0 & 1
          \end{array}
}
$$
and
$$
\bC\bG\bP(U_t) = \frac{t^2\ln t - (1-t)^2\ln(1-t)}{2(1-2t)},\quad
t\in[0,1].
$$
\begin{figure}[htbp]\centering
\includegraphics[width=0.5\textwidth]{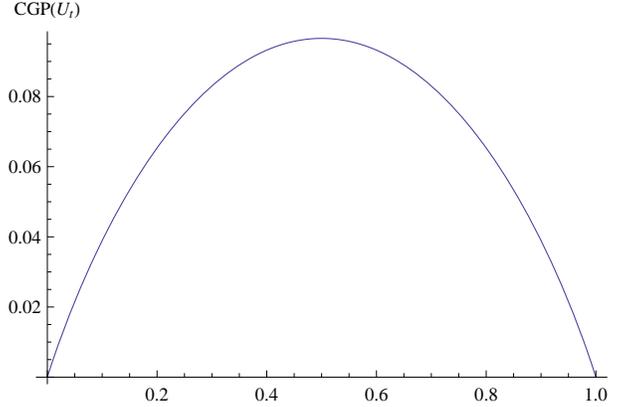}
\caption{The coherence generating power of $U_{t}$} \label{fig:Fig2}
\end{figure}
Again, we plot the $\bC\bG\bP(U_t)$ as the function of $t\in[0,1]$.
From Fig.~\ref{fig:Fig2}, we see that the maximal CGP of $U_{t}$,
attained at $t=0.5$, is given by
$\bC\bG\bP(U_{1/2})=\frac14(2\ln2-1)=0.097$, which is less than the
maximal CGP, $\ln 4-H_4+1=\ln 4-1/2-1/3-1/4=0.303$, of $4\times 4$
unitary matrices.
\end{exam}
\section{Conclusion}

Based on probabilistic averages, we have defined a measure of the
coherence generating power of a unitary operation: the average
coherence generated by the unitary channel acting on a uniform
ensemble of incoherent states. We have presented the explicit
analytical formula of CGP for any unitary channel and any finite
dimensions in terms of subentropy. An upper bound for CGP of a
unital quantum channel has been also derived. Detailed examples have
been studied.

We remark that Zanardi \emph{et al.} \cite{Zanardi1,Zanardi2}
studied the cohering and de-cohering power for unitary gates, based
on the coherence measure of Hilbert-Schmidt norm, which is not
really a well-defined measure of coherence. And their method is only
suitable for unital quantum channels since the Hilbert-Schmidt norm is
non-increasing under unital quantum channels. Hence the related
computation is relatively easy as it involves only integrals
in uniform Haar measure over pure states. In this work we
used the bona fide coherence measure of relative entropy. Our
approach applies to any quantum channels. The related computation
concerns complex integral techniques with Dirac delta function and
its Fourier integral representations. In addition, the
formula in \cite{Zanardi1,Zanardi2} for CGP of unitary channels
strongly depends on the dimension: the CGP approaches to zero when
the dimension increases. However, our CGP of any unitary channels
does not always approach to zero when the dimension goes to
infinite. It is generally very difficult to compute the CGP for
generic quantum channels. Our approach may highlight further
researches on such characterization of quantum coherence.

\subsubsection*{Acknowledgements}

This research was supported by Zhejiang Provincial Natural Science
Foundation of China under Grant No. LY17A010027 and NSFC
(Nos.11301124, 61771174), and also supported by the
cross-disciplinary innovation team building project of Hangzhou
Dianzi University. Other authors acknowledge supports from NSFC
Grant Nos. 11275131, 11571313 (ZM); No.11571313(ZC); No.
11675113(SF). Huangjun Zhu is acknowledged for helpful discussions.


\subsection*{Appendix A: About the proof of the Theorem}

We first introduce the following Lemma.

\begin{lem}[Jordan lemma]\label{lem:jordan}
Let $f(z)$ be analytic in the upper half-plane $\im(z)\geqslant0$,
except for a finite number of isolated points. Let also $C_R$ be an
arc of a semicircle $\abs{z}=R$ in the upper half-plane. If for each
$z$ on $C_R$, there is some constant $K_R$ such that
$\abs{f(z)}\leqslant K_R$ and $K_R\to0$ as $R\to\infty$, then for
$a>0$
\begin{eqnarray}
\lim_{R\to\infty} \int_{C_R} e^{\mathrm{i}az}f(z)\dif z = 0.
\end{eqnarray}
\end{lem}

\begin{proof}
Set $z=Re^{\mathrm{i}\theta}$ and take into account that
$\sin\theta\geqslant \frac2\pi\theta$ for
$0\leqslant\theta\leqslant\frac\pi2$. We have that, if $R\to\infty$,
\begin{eqnarray*}
\abs{\int_{C_R} e^{\mathrm{i}az}f(z)\dif z}&\leqslant& K_R\cdot
R\cdot\int^\pi_0 e^{-aR\sin\theta}\dif \theta\\
&\leqslant& K_R\frac\pi a\Pa{1-e^{-aR}}\to0.
\end{eqnarray*}
This completes the proof.
\end{proof}

If $a<0$ and $f(z)$ satisfies the conditions of Jordan lemma at
$\im(z)\leqslant0$, the formula is still valid but at the
integration over the arc $C_R$ in the lower half-plane. Similar
statements take place at $a=\pm\mathrm{i}\alpha(\alpha>0)$ if the
$C_R$-integration occurs in the right $(\Re(z)\geqslant0)$ or left
$(\Re(z)\leqslant0)$ half-plane, respectively.

Now we prove the following two formulae used in the proof of the
Theorem~\ref{th:Aone}.

(i).
\begin{eqnarray}\label{eq:alpha-int}
\cI_p(\alpha)=\frac{\Gamma(N)\Gamma(1-\alpha-N)}{\Gamma(-\alpha)}
\sum^N_{j=1}\frac{p^{\alpha+N-1}_j}{\prod_{i\neq j}(p_i-p_j)},
\end{eqnarray}

(ii).
\begin{eqnarray}\label{eq:prime-at-one}
\cI'_p(1) = -\frac1N\Pa{H_N-1+\bQ(p)},
\end{eqnarray}
where $\cI_p(\alpha)$ is given by (\ref{ipa}).

\begin{proof}
(i). From the Fourier transform of the Dirac delta function
$\delta$,
\begin{eqnarray}
\delta\Pa{1-\sum^N_{j=1}\lambda_j} =
\frac1{2\pi}\int^\infty_{-\infty}
\exp\Pa{\mathrm{i}t\Pa{1-\sum^N_{j=1}\lambda_j}}\dif t,
\end{eqnarray}
and the definition of Gamma function,
\begin{eqnarray}
\Pa{\sum^N_{j=1}p_j\lambda_j}^\alpha =\frac1{\Gamma(-\alpha)}
\int^\infty_0s^{-\alpha-1}\exp\Pa{-s\Pa{\sum^N_{j=1}p_j\lambda_j}}\dif
s,
\end{eqnarray}
we have
\begin{eqnarray}
\cI_p(\alpha)&=&
\Gamma(N)\int^\infty_0\cdots\int^\infty_0\Pa{\sum^N_{j=1}p_j\lambda_j}^\alpha
\delta\Pa{1-\sum^N_{j=1}\lambda_j}\prod^N_{k=1}\dif\lambda_k\\
&=&\frac{\Gamma(N)}{2\pi\Gamma(-\alpha)}\int^\infty_0\frac{\dif
s}{s^{\alpha+1}}\int^\infty_{-\infty}\dif t\cdot
e^{\mathrm{i}t}\Br{\int^\infty_0\cdots\int^\infty_0\prod^N_{k=1}\dif
\lambda_k \nabla_1\nabla_2},
\end{eqnarray}
where $\nabla_1=\exp\Pa{-s\Pa{\sum^N_{j=1}p_j\lambda_j}}$ and
$\nabla_2=\exp\Pa{-\mathrm{i}t\sum^N_{j=1}\lambda_j}$. Substituting
$f(x)=e^{-ax}H(x)$, where $H(x)=1_{(0,\infty)}$ is the Heaviside
step function and $a>0$, into the following formula,
\begin{eqnarray}
\widehat
f(\omega)=\frac1{\sqrt{2\pi}}\int^\infty_{-\infty}f(x)e^{-\mathrm{i}\omega
x}\dif x,
\end{eqnarray}
we obtain that
\begin{eqnarray}
\int^\infty_0e^{-ax}e^{-\mathrm{i}\omega x}\dif x =
\int^\infty_{-\infty}e^{-ax}H(x)e^{-\mathrm{i}\omega x}\dif x =
\frac1{\mathrm{i}\omega+a}.
\end{eqnarray}
Therefore
\begin{eqnarray}
&&\int^\infty_0\cdots\int^\infty_0\prod^N_{k=1}\dif \lambda_k
\exp\Pa{-s\Pa{\sum^N_{j=1}p_j\lambda_j}}
\exp\Pa{-\mathrm{i}t\sum^N_{j=1}\lambda_j}\notag\\
&&=\prod^N_{j=1}\int^\infty_0\dif \lambda_j
e^{-sp_j\lambda_j}e^{-\mathrm{i}t\lambda_j}
=\frac1{\prod^N_{j=1}(\mathrm{i}t+sp_j)}.
\end{eqnarray}
It follows that
\begin{eqnarray}
\cI_p(\alpha) =
\frac{\Gamma(N)}{2\pi\Gamma(-\alpha)}\int^\infty_0\frac{\dif
s}{s^{\alpha+1}}\Set{\int^\infty_{-\infty}
\frac{e^{\mathrm{i}t}}{\prod^N_{j=1}(\mathrm{i}t+sp_j)}\dif t}.
\end{eqnarray}
By using complex integral techniques in Lemma~\ref{lem:jordan}, we
get
\begin{eqnarray}
\int^\infty_{-\infty}
\frac{e^{\mathrm{i}t}}{\prod^N_{j=1}(\mathrm{i}t+sp_j)}\dif t =
\frac{2\pi}{s^{N-1}}\sum^N_{j=1}\frac{e^{-sp_j}}{\prod_{i \neq
j}(p_i-p_j)},
\end{eqnarray}
which gives rise to
\begin{eqnarray}
\cI_p(\alpha) &=&
\frac{\Gamma(N)}{\Gamma(-\alpha)}\int^\infty_0\frac{\dif
s}{s^{\alpha+N}}\sum^N_{j=1}\frac{e^{-sp_j}}{\prod_{i\neq j}(p_i-p_j)}\\
&=&\frac{\Gamma(N)}{\Gamma(-\alpha)}\sum^N_{j=1}\frac1{\prod_{i\neq
j}(p_i-p_j)}\int^\infty_0s^{-\alpha-N}e^{-sp_j}\dif
s\\
&=&\frac{\Gamma(N)\Gamma(1-\alpha-N)}{\Gamma(-\alpha)}\sum^N_{j=1}\frac{p^{\alpha+N-1}_j}{\prod_{i\neq
j}(p_i-p_j)}.
\end{eqnarray}

(ii). From the property of the Gamma function:
\begin{eqnarray}
\Gamma(1-z)\Gamma(z)=\frac{\pi}{\sin(\pi z)},
\end{eqnarray}
we have
\begin{eqnarray}
\Gamma(1-\alpha-N) =
\frac{\pi}{\Gamma(\alpha+N)\sin(\alpha\pi+N\pi)}
\end{eqnarray}
and
\begin{eqnarray}
\Gamma(-\alpha) = \frac{\pi}{\Gamma(\alpha+1)\sin(\alpha\pi+\pi)}.
\end{eqnarray}
Therefore $\cI_p(\alpha)$ can be rewritten as
\begin{eqnarray}
\cI_p(\alpha)&=&\frac{\Gamma(N)\Gamma(\alpha+1)}{\Gamma(\alpha+N)}\frac{\sin(\alpha\pi+N\pi)}{\sin(\alpha\pi+\pi)}
\sum^N_{j=1}\frac{p^{\alpha+N-1}_j}{\prod_{i\neq j}(p_i-p_j)}\\
&=&(-1)^{N-1}\frac{\Gamma(N)\Gamma(\alpha+1)}{\Gamma(\alpha+N)}\sum^N_{j=1}\frac{p^{\alpha+N-1}_j}{\prod_{i\neq
j}(p_i-p_j)},
\end{eqnarray}
which gives rise to
\begin{eqnarray}
\cI_p(\alpha)=\frac{\Gamma(N)\Gamma(\alpha+1)}{\Gamma(\alpha+N)}\sum^N_{j=1}\frac{p^{\alpha+N-1}_j}{\prod_{i\neq
j}(p_j-p_i)}.
\end{eqnarray}
Taking the derivative of $\cI_p(\alpha)$ with respect to $\alpha$,
we get
\begin{eqnarray}
\cI'_p(\alpha)&=&\displaystyle\Gamma(N)\frac{\dif}{\dif\alpha}\frac{\Gamma(\alpha+1)}{\Gamma(\alpha+N)}\sum^N_{j=1}\frac{p^{\alpha+N-1}_j}{\prod_{i\neq
j}(p_j-p_i)}\notag\\
&&
+\frac{\Gamma(N)\Gamma(\alpha+1)}{\Gamma(\alpha+N)}\sum^N_{j=1}\frac{p^{\alpha+N-1}_j\ln
p_j}{\prod_{i\neq j}(p_j-p_i)}.
\end{eqnarray}
This implies that, for $\alpha=1$,
\begin{eqnarray}\label{pp}
\cI'_p(1)=\frac1N(\psi(2)-\psi(1+N))\sum^N_{j=1}\frac{p^N_j}{\prod_{i\neq
j}(p_j-p_i)} + \frac1N \sum^N_{j=1}\frac{p^N_j\ln p_j}{\prod_{i\neq
j}(p_j-p_i)},
\end{eqnarray}
where $\psi(2)=1-\gamma$, $\psi(1+N)=H_N-\gamma$, where
$\gamma\simeq0.57721$.

We compute the following summation in \eqref{pp},
\begin{eqnarray}
F(p_1,\ldots,p_N):=\sum^N_{j=1}\frac{p^N_j}{\prod_{i\neq
j}(p_j-p_i)}.
\end{eqnarray}
Since it is a rational symmetric function, homogeneous of degree
one, with all singularities removable, it must be a multiple of
$\sum^N_{j=1}p_j$. That is, $F(tp_1,\ldots,tp_N)=tF(p_1,\ldots,p_N)$
for any real number $t$, and
$F(p_{\sigma(1)},\ldots,p_{\sigma(N)})=F(p_1,\ldots,p_N)$ for all
permutations $\sigma \in S_N$. This means that
\begin{eqnarray}
F(p_1,\ldots,p_N)\propto \sum^N_{j=1}p_j.
\end{eqnarray}
Without loss of generality, assume that
$F(p_1,\ldots,p_N)=C\cdot\sum^N_{j=1}p_j$ for some constant $C$. By
setting $(p_1,\ldots,p_N)=(1,0,\ldots,0)$, we get $C=1$. That is,
\begin{eqnarray}
\sum^N_{j=1}\frac{p^N_j}{\prod_{i\neq j}(p_j-p_i)} = \sum^N_{j=1}p_j
= 1.
\end{eqnarray}
Therefore, from \eqref{se}, \eqref{pp} gives rise to
\begin{eqnarray}
-N\cdot \cI'_p(1) &=&\displaystyle (\psi(1+N) -
\psi(2))\sum^N_{j=1}\frac{p^N_j}{\prod_{i\neq j}(p_j-p_i)}
-\sum^N_{j=1}\frac{p^N_j\ln p_j}{\prod_{i\neq
j}(p_j-p_i)}\\
&=& H_N-1 + \bsQ(p).
\end{eqnarray}
Hence $\cI'_p(1) = -\frac1N\Pa{H_N-1+\bsQ(p)}$.
\end{proof}

\subsection*{Appendix B: Proof of $\bsQ(B^\t) \leqslant \bsH(B)$}

To prove the relation $\bsQ(B^\t) \leqslant \bsH(B)$, we prove that
following relation first:
\begin{eqnarray}\label{r4}
\Gamma(N)\int\bsH_\lambda(B)\delta\Pa{1-\sum^N_{i=1}\lambda_i}\prod^N_{j=1}\dif\lambda_j=\frac1N\Pa{\sum^N_{i=1}\bsH(\beta_i)}.
\end{eqnarray}

\begin{proof}
Since $\bsH(B\lambda)-\bsH(\lambda)\leqslant
\bsH_\lambda(B)=\sum^N_{j=1}\lambda_j\bsH(\beta_j)$, it follows that
\begin{eqnarray}
&&\Gamma(N)\int\bsH_\lambda(B)\delta\Pa{1-\sum^N_{i=1}\lambda_i}\prod^N_{j=1}\dif\lambda_j\\
&&=\Gamma(N)\int\Pa{\sum^N_{i=1}\lambda_i\bsH(\beta_i)}\delta\Pa{1-\sum^N_{i=1}\lambda_i}\prod^N_{j=1}\dif\lambda_j\\
&&=\Gamma(N)\sum^N_{i=1}\bsH(\beta_i)\int\lambda_i\delta\Pa{1-\sum^N_{i=1}\lambda_i}\prod^N_{j=1}\dif\lambda_j\\
&&=\Pa{\sum^N_{i=1}\bsH(\beta_i)}\Gamma(N)\int\lambda_1\delta
\Pa{1-\sum^N_{i=1}\lambda_i}\prod^N_{j=1}\dif\lambda_j.
\end{eqnarray}
Denote
\begin{eqnarray}
f(t) =
\Gamma(N)\int\lambda_1\delta\Pa{t-\sum^N_{i=1}\lambda_i}\prod^N_{j=1}\dif\lambda_j.
\end{eqnarray}
Performing Laplace transform ($t\to s$) of $f$, we obtain
\begin{eqnarray}
\widetilde f(s)=\int^\infty_0 f(t)e^{-st}\dif t
=\Gamma(N)\int\prod^N_{j=1}\dif\lambda_j\Pa{\lambda_1\int^\infty_0
\delta\Pa{t-\sum^N_{i=1}\lambda_i} e^{-st}\dif t}.
\end{eqnarray}
That is,
\begin{eqnarray}
\widetilde f(s) &=&
\Gamma(N)\int\prod^N_{j=1}\dif\lambda_j\Pa{\lambda_1\int^\infty_0
\delta\Pa{t-\sum^N_{i=1}\lambda_i} e^{-st}\dif t}\\
&=&
\Gamma(N)\int\prod^N_{j=1}\dif\lambda_j\lambda_1e^{-s\sum^N_{i=1}\lambda_i}\\
&=&\Gamma(N)\int\lambda_1 e^{-s\lambda_1}\dif\lambda_1\times\int
e^{-s\lambda_2}\dif\lambda_2\times\cdots\times\int
e^{-s\lambda_N}\dif\lambda_N\\
&=&\Gamma(N)s^{-N-1}\int^\infty_0 xe^{-x}\dif
x=\frac{\Gamma(N)}{s^{N+1}}.
\end{eqnarray}
Thus $f(t)=\frac1Nt^N$. Therefore
\begin{eqnarray}
\Gamma(N)\int\bsH_\lambda(B)\delta\Pa{1-\sum^N_{i=1}\lambda_i}\prod^N_{j=1}\dif\lambda_j=\frac1N\Pa{\sum^N_{i=1}\bsH(\beta_i)}.
\end{eqnarray}
We have done.
\end{proof}

As a by-product of the formula \eqref{r4}, we have $\bsQ(B^\t)
\leqslant \bsH(B)$.


\begin{thebibliography}{}


\bibitem{Mandel}
L. Mandel and E. Wolf, Optical Coherence and Quantum Optics
(Cambridge University Press, Cambridge, England, 1995).

\bibitem{London}
F. London and H. London,
\href{https://doi.org/10.1098/rspa.1935.0048}{Proc. R. Soc. A
\textbf{149}, 71 (1935).}

\bibitem{Horodecki}
P. \'{C}wikli\'{n}ski, M. Studzi\'{n}ski, M. Horodecki, and J.
Oppenheim,
\href{https://doi.org/10.1103/PhysRevLett.115.210403}{\prl
\textbf{115}, 210403 (2015).}

\bibitem{Bagan}
E. Bagan, J. A. Bergou, S. S. Cottrell, and M. Hillery,
\href{https://doi.org/10.1103/PhysRevLett.116.160406}{\prl
\textbf{116}, 160406 (2016).}

\bibitem{Jha}
P. K. Jha, M. Mrejen, J. Kim, C. Wu, Y. Wang, Y. V. Rostovtsev, and
X. Zhang, \href{https://doi.org/10.1103/PhysRevLett.116.165502}{\prl
\textbf{116}, 165502 (2016).}

\bibitem{Kammerlander}
P. Kammerlander and J. Anders,
\href{https://doi.org/10.1038/srep22174}{Sci. Rep. \textbf{6}, 22174
(2016).}

\bibitem{Shi}
H. L. Shi, S. Y. Liu, X. H. Wang, W. L. Yang, Z. Y. Yang, H. Fan,
\href{https://doi.org/10.1103/PhysRevA.95.032307}{\pra \textbf{95},
032307 (2017).}

\bibitem{Lloyd}
S. Lloyd,
\href{https://doi.org/10.1088/1742-6596/302/1/012037}{J.Phys.: Conf.
Ser. \textbf{302}, 012037 (2011).}

\bibitem{Li}
C. M. Li, N. Lambert, Y. N. Chen, G. Y. Chen, and F. Nori,
\href{https://doi.org/10.1038/srep00885}{Sci. Rep. \textbf{2}, 885
(2012).}

\bibitem{Huelga}
S. Huelga and M. Plenio,
\href{https://doi.org/10.1080/00405000.2013.829687}{Contemporary
Physics \textbf{54}, 181 (2013).}

\bibitem{Singh0}
V. Singh Poonia, D. Saha, and S. Ganguly,
\href{https://arxiv.org/abs/1408.1327}{arXiv:1408.1327, (2014).}

\bibitem{Baumgratz}
T. Baumgratz, M. Cramer, and M. B. Plenio,
\href{https://doi.org/10.1103/PhysRevLett.113.140401}{\prl
\textbf{113}, 140401 (2014).}

\bibitem{Theurer2017}
T. Theurer, N. Killoran, D. Egloff, and M.B. Plenio,
\href{https://doi.org/10.1103/PhysRevLett.119.230401}{\prl
\textbf{119}, 230401 (2017).}

\bibitem{Aberg2006}
J. Aberg,
\href{https://arxiv.org/abs/quant-ph/0612146}{arXiv:quant-ph/0612146}


\bibitem{Winter}
A. Winter and D. Yang,
\href{https://doi.org/10.1103/PhysRevLett.116.120404}{\prl
\textbf{116}, 120404 (2016).}

\bibitem{Cheng}
S. Cheng and M. J. W. Hall,
\href{https://doi.org/10.1103/PhysRevA.92.042101}{\pra \textbf{92},
042101 (2015).}

\bibitem{ROC}
C. Napoli, T. R. Bromley, M. Cianciaruso, M. Piani, N. Johnston, and
G. Adesso,
\href{https://doi.org/10.1103/PhysRevLett.116.150502}{\prl
\textbf{116}, 150502 (2016).}

\bibitem{Singh}
U. Singh, L. Zhang, and A. K. Pati,
\href{https://doi.org/10.1103/PhysRevA.93.032125}{\pra \textbf{93},
032125 (2016).}


\bibitem{aop2017}
L. Zhang, U. Singh, A.K. Pati,
\href{https://doi.org/10.1016/j.aop.2016.12.024}{\aop \textbf{377},
125-146 (2017).}

\bibitem{bukf}
K. Bu, U. Singh, S-M. Fei, A. K. Pati, J. Wu,
\href{https://doi.org/10.1103/PhysRevLett.119.150405}{\prl
\textbf{119}, 150405 (2017).}

\bibitem{Marvian2016}
I. Marvian and R.W. Spekkens,
\href{https://doi.org/10.1103/PhysRevA.94.052324}{\pra \textbf{94},
052324 (2016).}


\bibitem{Piani}
M. Piani, M. Cianciaruso, T. R. Bromley, C. Napoli, N. Johnston, and
G. Adesso, \href{https://doi.org/10.1103/PhysRevA.93.042107}{\pra
\textbf{93}, 042107 (2016).}

\bibitem{Marvian}
I. Marvian, R. W. Spekkens, and P. Zanardi,
\href{https://doi.org/10.1103/PhysRevA.93.052331}{\pra \textbf{93},
052331 (2016).}

\bibitem{Streltsov}
A. Streltsov, U. Singh, H. S. Dhar, M. N. Bera, and G. Adesso,
\href{https://doi.org/10.1103/PhysRevLett.115.020403}{\prl
\textbf{115}, 020403 (2015).}



\bibitem{Radhakrishnan}
C. Radhakrishnan, M. Parthasarathy, S. Jambulingam, and T. Byrnes,
\href{https://doi.org/10.1103/PhysRevLett.116.150504}{\prl
\textbf{116}, 150504 (2016).}

\bibitem{Ma}
J. Ma, B. Yadin, D. Girolami, V. Vedral, and M. Gu,
\href{https://doi.org/10.1103/PhysRevLett.116.160407}{\prl
\textbf{116}, 160407 (2016).}

\bibitem{Karpat}
G. Karpat, B. Cakmak, and F. F. Fanchini,
\href{https://doi.org/10.1103/PhysRevB.90.104431}{Phys. Rev. B
\textbf{90}, 104431 (2014).}

\bibitem{Malvezzi}
A. L. Malvezzi, G. Karpat, B. Cakmak, F. F. Fanchini, T. Debarba,
and R. O. Vianna,
\href{https://doi.org/10.1103/PhysRevB.93.184428}{Phys. Rev. B
\textbf{93}, 184428 (2016).}


\bibitem{Chitambar}
E. Chitambar, M. H. Hsieh,
\href{https://doi.org/10.1103/PhysRevLett.117.020402}{\prl
\textbf{117}, 020402 (2016).}


\bibitem{Zhu}
H. J. Zhu, Z. H. Ma, Z. Cao, S. M. Fei, V. Vedral,
\href{https://doi.org/10.1103/PhysRevA.96.032316}{\pra \textbf{96},
032316 (2017).}




\bibitem{Mani}
A. Mani and V. Karimipour,
\href{https://doi.org/10.1103/PhysRevA.92.032331}{\pra \textbf{92},
032331 (2015).}

\bibitem{Misra}
A. Misra, U. Singh, S. Bhattacharya, and A. K. Pati,
\href{https://doi.org/10.1103/PhysRevA.93.052335}{\pra \textbf{93},
052335 (2016).}

\bibitem{MGD}
M. G. D\'{i}az, D. Egloff, M.B. Plenio, Quant. Inf. Comput.
\textbf{16}, 1282-1294 (2016).

\bibitem{Zanardi1}
P. Zanardi, G. Styliaris, and L.C. Venuti,
\href{https://doi.org/10.1103/PhysRevA.95.052306}{\pra \textbf{95},
052306(2017).}

\bibitem{Zanardi2}
P. Zanardi, G. Styliaris, and L. C. Venuti,
\href{https://doi.org/10.1103/PhysRevA.95.052307}{\pra \textbf{95},
052307 (2017).}








\bibitem{LJ2011}
L. Zhang and J. Wu,
\href{https://doi.org/10.1016/j.physleta.2011.10.008}{\pla
\textbf{375}, 4163-4165 (2011).}


\bibitem{Lasota1994}
A. Lasota and M.C. Mackey, \newblock{\em Chaos, Fractals, and Noise.
Stochastic Aspects of Dynamics}, Springer-Verlag, New York (1994).

\bibitem{Page1993}
D.N. Page, \href{https://doi.org/10.1103/PhysRevLett.71.1291}{\prl
\textbf{71}, 1291 (1993).}

\bibitem{lin2017}
L. Zhang, \href{https://doi.org/10.1088/1751-8121/aa6179}{\jpa:
Math. Theor. \textbf{50}, 155303 (2017).}


\bibitem{buscemi2016pra}
F. Buscemi, S. Das, M.M. Wilde,
\href{https://doi.org/10.1103/PhysRevA.93.062314}{\pra \textbf{93},
062314 (2016).}




































\bibitem{Audenaert2016}
K.M.R. Audenaert, N. Datta, M. Ozols,
\href{https://doi.org/10.1063/1.4950785}{\jmp \textbf{57}, 052202
(2016).}






\end{thebibliography}
\end{document}